\documentclass[11pt]{article}

\usepackage{geometry}
\usepackage{cite}
\usepackage{paralist}
\usepackage{graphicx}
\usepackage{subfigure}
\usepackage{amssymb}
\usepackage{amsmath}
\allowdisplaybreaks[1]

\usepackage{amsthm}

\newtheorem{theorem}{Theorem}
\newtheorem{lemma}[theorem]{Lemma}

\makeatletter
\def\@endtheorem{\endtrivlist}
\makeatother

\newcounter{rules}
\newenvironment{Rule}{\refstepcounter{rules}\par\smallskip\noindent
\textbf{(\arabic{rules})}\quad}{\par} %\par\noindent\rule{\textwidth}{1pt}}
\newcommand{\currentrule}{\arabic{rules}}

\newenvironment{Ruleb}[1]{\par\smallskip\noindent
\textbf{(#1)}\quad}{}

\newcommand{\branchb}{\textsc{Branch}}
\newcommand{\vcb}{\textsc{VC}}
\newcommand{\vcbaseb}{\textsc{VCbase}}
\newcommand{\branchtb}{\textsc{Branch}\ensuremath{{}_3}}
\newcommand{\vctb}{\textsc{VC}\ensuremath{{}_3}}
\newcommand{\vctbaseb}{\textsc{VCbase}\ensuremath{{}_3}}

\newcommand{\branch}[3]{\branchb(#1,#2,\allowbreak\{#3\})}

\newcommand{\brancht}[3]{\branchtb(#1,#2,\allowbreak\{#3\})}
\newcommand{\vct}[2]{\vctb(#1,#2)}
\newcommand{\vctbase}[2]{\vctbaseb(#1,#2)}

\begin{document}

\title{Above guarantee parameterization for vertex cover on graphs with maximum degree 4}
\author{Dekel Tsur%
\thanks{Ben-Gurion University of the Negev.
Email: \texttt{dekelts@cs.bgu.ac.il}}}
\date{}
\maketitle

\begin{abstract}
In the vertex cover problem, the input is a graph $G$ and an integer $k$, and
the goal is to decide whether there is a set of vertices $S$ of size
at most $k$
such that every edge of $G$ is incident on at least one vertex in $S$.
We study the vertex cover problem on graphs with maximum degree 4 and
minimum degree at least 2, parameterized by $r = k-n/3$.
We give an algorithm for this problem whose running time is $O^*(1.6253^r)$.
As a corollary, we obtain an $O^*(1.2403^k)$-time algorithm for vertex
cover on graphs with maximum degree~4.
\end{abstract}


\section{Introduction}
For an undirected graph $G$, a \emph{vertex cover} of $G$ is a set of vertices
$S$ such that every edge of $G$ is incident on at least one vertex in $S$.
In the parameterized \emph{vertex cover problem},
the input is a graph $G$ and an integer $k$
and the goal is to decide whether there is a vertex cover of $G$ of size
at most $k$.
The parameterized vertex cover problem has been studied extensively.
The first parameterized algorithm for vertex cover was given
in~\cite{buss1993nondeterminism}.
Improved algorithms were given in~\cite{downey1995parameterized, %
balasubramanian1998improved,downey1999parameterized,niedermeier1999upper, %
stege1999improved,chen2001vertex,chen2010improved, %
niedermeier2003efficient,chandran2004refined}.
The parameterized vertex cover problem was also studied on graphs with maximum
degree 3~\cite{chen2001vertex,chen2000improvement,chen2005labeled,razgon2009faster,xiao2010note}
and maximum degree 4~\cite{chen2001vertex,agrawal2014vertex}.

%The fastest parameterized algorithm for vertex cover is due to
%Chen et al.~\cite{chen2010improved}.
%Agrawal et al.~\cite{agrawal2014vertex} gave an $O^*(1.264^k)$-time algorithm
%for vertex cover on graphs with maximum degree~4.

Consider the vertex problem on graphs with maximum degree $\Delta$.
It is easy to eliminate vertices with degree at most $1$ from the graph.
Thus, we can assume that the input graph has minimum degree at least $2$.
It is easy to show that the minimum size of a vertex cover in a graph with
maximum degree $\Delta$ and minimum degree at least $2$ is at least
$\frac{2}{2+\Delta}n$, where $n$ is the number of vertices
(cf.~\cite{xiao2010note}).
Therefore, it is more natural to use the ``above guarantee'' parameter
$r_\Delta = k-\frac{2}{2+\Delta}n$.
For $\Delta = 3$, Xiao~\cite{xiao2010note} gave an algorithm with
$O^*(1.6651^{r_3})$ running time.
An $O(c^{r_\Delta})$-time algorithm for vertex cover can also give an algorithm
for vertex cover parameterized by $k$.
Suppose that there is an exponential algorithm for vertex cover on graphs with
maximum degree $\Delta$ whose running time is $O^*(d^n)$.
Let $\alpha = 1/(2/(2+\Delta) + \log_c d)$.
Given an instance $(G,k)$ of vertex cover,
if $n > \alpha k$, run the parameterized algorithm on $(G,k)$ in
$O^*(c^{r_\Delta}) = O^*(d^{\alpha k})$ time,
and otherwise run the exponential algorithm on $G$ in
$O^*(d^n) = O^*(d^{\alpha k})$ time.
For graphs with maximum degree~3, combining the $O^*(1.6651^{r_3})$-time
algorithm of Xiao~\cite{xiao2010note} with the $O(1.0836^n)$-time algorithm
of Xiao and Nagamochi~\cite{xiao2013confining} gives an $O^*(1.1558^k)$-time
algorithm, which is faster than the previous algorithms given for this
problem~\cite{chen2001vertex,chen2000improvement,chen2005labeled,razgon2009faster}.

In this paper we give an algorithm for vertex cover on graphs with maximum
degree 4 and minimum degree at least 2 whose running time is
$O^*(1.6253^{r_4})$.
Combining our algorithm with the $O(1.1376^n)$-time algorithm of
Xiao and Nagamochi~\cite{xiao2017refined} gives an $O^*(1.2403^k)$-time
algorithm for vertex cover on graphs with maximum degree~4.
This algorithm is faster than the $O^*(1.2637^k)$-time algorithm of
Agrawal et al.~\cite{agrawal2014vertex}.

Our algorithm is similar to the algorithm of Xiao~\cite{xiao2010note}
for graphs with maximum degree~3.
Unlike the algorithm of Xiao that uses one branching rule, our algorithm
uses several branching rules.
The branching rules of our algorithm are based on the rules of the algorithm of
Chen et al.~\cite{chen2001vertex} with minor modifications in order obtain a faster algorithm.
Since the rules of our algorithm are based on the rules
of~\cite{chen2001vertex},
% and other previous work,
we omit the proof of correctness of these rules.

\section{Preliminaries}
For a graph $G$ and a set of vertices $S$, $G-S$ is the graph obtained from
$G$ by removing the vertices of $S$ and the edges incident on these vertices.
For a set of vertices $S$, $N(S) = (\bigcup_{v\in S} N(v))\setminus S$.

The \emph{merge} operation on a set of vertices $S$ in a graph $G$ generates a
graph $G'$ by deleting the vertices of $S$ and adding a new vertex $v^*$.
The vertex $v^*$ is adjacent to a vertex $u$ in $G'$ if and only if
there is an edge in $G$ between $u$ and a vertex in $S$.

\subsection{Reduction rules}
In this section we describe several reduction rules that are used by our
algorithm.

Rule~(D0) and Rule~(D1) below handle vertices with degree at most 1.
While we assume that the input graph has minimum degree 2,
graphs generated during the algorithm can contain vertices with degree 
at most~1 and these rules are used to eliminate such vertices.
\begin{Ruleb}{D0}
If there is a vertex $v$ with degree~0, delete $v$ from $G$.
\end{Ruleb}

\begin{Ruleb}{D1}
If there is a vertex $v$ with degree~1, delete the unique neighbor of $v$
from $G$ and decrease $k$ by $1$.
\end{Ruleb}

%The following rule is called the \emph{folding rule}.
%\begin{Ruleb}{F}
%Let $v$ be a vertex with degree $2$ such that the two neighbors $u,w$ of $v$
%are not adjacent. Merge the vertices $u$ and $w$,  delete $v$, and decrease $k$
%by $1$.
%\end{Ruleb}

A \emph{general crown} in a graph $G$ is a pair $C,H$ of disjoint nonempty
sets of vertices such that the vertices in $C$ have degree~0 in $G-H$.
A crown $C,H$ is called \emph{good} if $G-C-H$ has minimum degree at least~2.
A \emph{proper crown} is a general crown $C,H$ such that
there is a matching between $C$ and $H$ of size $|H|$.
An \emph{almost crown} is a general crown $C,H$ such that
$|H|=|C|+1$ and $|N(S)| \geq |S|+1$ for every $\emptyset\neq S\subseteq C$.
We use the following reduction rules from~\cite{chen2010improved}.

\begin{Ruleb}{C1}
If $C,H$ is a proper crown, delete $C\cup H$ from $G$ and decrease $k$ by $|H|$.
\end{Ruleb}
\begin{Ruleb}{C2}
If $C,H$ is an almost crown and $H$ is not an independent set,
delete $C\cup H$ from $G$ and decrease $k$ by $|H|$.
\end{Ruleb}
\begin{Ruleb}{C3}
If $C,H$ is an almost crown and $H$ is an independent set,
merge the vertices in $C \cup H$ and decrease $k$ by $|H|-1$.
\end{Ruleb}

\begin{lemma}[Xiao~\cite{xiao2010note}]\label{lem:crown}
Let $G$ be a graph with minimum degree at least 2.
If $C,H$ is a general crown and $|C| \geq |H|-1$,
it is possible in polynomial time to either find a good proper crown $C',H'$,
or to conclude that $C,H$ is an almost crown.
\end{lemma}

\section{The Algorithm}

In this section we give an algorithm for vertex cover parameterized by $r = r_4$
on graphs with maximum degree at most 4 and minimum degree at least 2.
We first describe an algorithm \vctbaseb\ for solving the vertex cover problem
on \emph{connected} graphs with maximum degree at most 3 and
minimum degree at least 2.
The algorithm maintains the following invariants on the current graph:
\begin{inparaenum}[(1)]
\item
The minimum degree is at least 2.
\item
The maximum degree is at most 3.
\item
Except for the initial instance, every connected component is not 3-regular.
\end{inparaenum}

Before describing algorithm \vctbaseb, we give a procedure \branchtb.
The input to the procedure is an instance $(G,k)$ of vertex cover and
sets of vertices $S_1,\ldots,S_t$.
The goal of the procedure is to performs branching on the instances
$(G-S_1,k-|S_1|),\allowbreak\ldots,\allowbreak(G-S_t,k-|S_t|)$, after applying Rule~(D0)
and Rule~(D1) on each instance.
However, it is possible that due to many applications of Rule~(D0),
the decrease in $r$ in at least one branch will be too small.
In this case, the procedure performs a suitable reduction rule, or
make a recursive call. When procedure \branchtb\ makes a recursive call,
it marks one of the sets.
Procedure \branchtb\ performs the following steps.
\begin{enumerate}
\item\label{branch:apply-D0-D1}
For $i=1,\ldots,t$,
repeatedly apply Rule~(D0) and Rule~(D1) on the instance $(G-S_i,k-|S_i|)$.
Let $(G_i,k_i)$ be resulting instance.
Let $C_i$ be a set containing all the vertices that were
deleted by applications of Rule~(D0).
Let $H_i$ be a set containing $S_i$ and all the vertices that were
deleted by applications of Rule~(D1)
(note that $(G_i,k_i) = (G-C_i-H_i,k-|H_i|)$).
\item\label{branch:reduce}
If there is $i$ such that $S_i$ is not marked, $C_i \neq \emptyset$,
$H_i \neq \emptyset$, and $|C_i| \geq |H_i|-1$:
\begin{enumerate}
\item\label{branch:crown}
Apply Lemma~\ref{lem:crown} on the
general crown $C_i,H_i$ and either obtain a good proper crown $C,H$
in $G$, or conclude that $C_i,H_i$ is an almost crown in $G$.
\item\label{branch:C1}
If $C,H$ is a proper crown, apply Rule~(C1) on $(G,k)$ and $C,H$,
and let $(G',k')$ be the resulting instance.
Return $\vctbase{G'}{k'}$.
\item\label{branch:C2}
If $H_i$ is not an independent set, apply Rule~(C2) on $(G,k)$ and $C_i,H_i$,
and let $(G',k')$ be the resulting instance.
Return $\vctbase{G'}{k'}$.
\item\label{branch:C3}
If $|N(H_i)\setminus C_i| \leq 2$, apply Rule~(C3) on $(G,k)$ and $C_i,H_i$,
and let $(G',k')$ be the resulting instance.
Return $\vctbase{G'}{k'}$.
%\marginpar{Note that D0,D1 cannot be applied on $G'$}
\item\label{branch:recursion}
Otherwise, call $\brancht{G}{k}{H_i,N(H_i)}$, where $H_i$ is marked.

\end{enumerate}
\item\label{branch:branch}
Return $\vctbase{G_1}{k_1} \lor \cdots \lor \vctbase{G_t}{k_t}$.
\end{enumerate}
Note that the graph $G'$ in line~\ref{branch:C1}, \ref{branch:C2},
or~\ref{branch:C3} has minimum degree at least~2 (as $C,H$ and $C_i,H_i$ are
good crowns).
%Note that if line~\ref{enu:reduction} is executed,
%procedure \branchtb\ performs only one call to \vctbaseb.
%Otherwise, the procedure performs $t$ calls to \vctbaseb.

%We note that in the proof of the lemma we assume that the maximum degree
%of $G$ is at most 4, as we also need this lemma for the algorithm for graphs
%with maximum degree~4.
\begin{lemma}\label{lem:branch-C123}
If line~\ref{branch:C1}, \ref{branch:C2}, or~\ref{branch:C3} of procedure
\branchtb\ is executed,
the value of $r$ does not increase when moving from the
instance $(G,k)$ to the instance $(G',k')$.
\end{lemma}
\begin{proof}
Suppose first that line~\ref{branch:C1} is executed.
Since $G$ has minimum degree at least~2, the number of edges between
$C$ and $H$ is at least $2|C|$.
Each vertex in $H$ has degree at most 4 (we note that we used 4 and not 3 since
we will need this lemma also for the algorithm on graphs with maximum degree~4),
and therefore
the number of edges between $C$ and $H$ is at most $4|H|$.
It follows that $2|C|\leq 4|H|$.
By definition,
$k' = k - |H|$ and $|V(G')| = |V(G)|-|C|-|H|$.
Therefore, the value of $r$ decreases by
$|H|-(|C|+|H|)/3 = (2|H|-|C|)/3 \geq 0$.
%\end{proof}

%\begin{lemma}\label{lem:branch-C2}
%If line~\ref{branch:C2} of procedure \branchtb\ is executed,
%the value of $r$ decreases by at least $1$ when moving from the
%instance $(G,k)$ to the instance $(G',k')$.
%\end{lemma}
%\begin{proof}
Now suppose that line~\ref{branch:C2} is executed.
By definition, $k' = k - |H_i|$ and
$|V(G')| = |V(G)|-(|C_i|+|H_i|) = |V(G)|-(2|H_i|-1)$.
Therefore, the value of $r$ decreases by
$|H_i|-(2|H_i|-1)/3 = (|H_i|+1)/3 \geq 1$.
%\end{proof}
%
%\begin{lemma}\label{lem:branch-C3}
%If line~\ref{branch:C3} of procedure \branchtb\ is executed,
%the value of $r$ decreases by at least $\frac{1}{3}$ when moving from the
%instance $(G,k)$ to the instance $(G',k')$.
%\end{lemma}
%\begin{proof}
Finally, if line~\ref{branch:C3} is executed,
$k' = k - (|H_i|-1)$ and
$|V(G')| = |V(G)|-(|C_i|+|H_i|-1) = |V(G)|-2(|H_i|-1)$.
Therefore, the value of $r$ decreases by
$|H_i|-1-2(|H_i|-1)/3 = (|H_i|-1)/3 \geq \frac{1}{3}$.
\end{proof}

\begin{lemma}\label{lem:r-branch}
Suppose that line~\ref{branch:branch} of procedure \branchtb\ is executed.
For every $i$ such that $S_i$ is not marked,
when moving from the instance $(G,k)$ to the instance $(G_i,k_i)$,
the value of $r$ decreases by at least $(|H_i|+2)/3$ if $|H_i| \geq 2$,
decreases by $\frac{2}{3}$ if $|H_i|=1$, and does not change if $|H_i| = 0$.
\end{lemma}
\begin{proof}
Suppose that $|H_i| \geq 2$.
We have that $|C_i| \leq |H_i| - 2$ otherwise line~\ref{branch:branch} of the
algorithm would not be executed.
By definition, in the instance $(G_i,k_i)$ the value of $k$ decreases
by $|H_i|$ and the value of $n$ decreases by $|H_i|+|C_i|$.
Therefore, the value of $r$ decreases by $|H_i|-(|H_i|+|C_i|)/3 =
(2|H_i|-|C_i|)/3 \geq (|H_i|+2)/3$.

If $|H_i| = 1$ then $|C_i| = 0$ otherwise line~\ref{branch:branch} of the
algorithm would not be executed.
Therefore, the value of $r$ decreases by $1-1/3 = \frac{2}{3}$.
Finally, if $|H_i| = 0$ then $|C_i| = 0$, and the value of $r$ does not change.
\end{proof}

We now analyze a call to \branchtb\ in which $S_1$ is marked and 
line~\ref{branch:branch} is executed.
By definition, $S_1 = H'$ and $S_2 = N(H')$, where $C',H'$ is
an almost crown that was obtained during the parent call to \branchtb.
Note that since Rules~(D0) and~(D1) cannot be applied on $G-C'-H'$, it follows
that $H_1 = S_1 = H'$ and $C_1 = C'$.
When moving from $(G,k)$ to $(G_1,k_1)$, the value
of $k$ decreases by $|H_1|=|H'|$ and the value of $n$ decreases by
$|H_1|+|C_1| = 2|H'|-1$. Therefore, the value of $r$ decreases by
$|H'|-(2|H'|-1) = (|H'|+1)/3 \geq 1$.
By Lemma~\ref{lem:r-branch}, when moving from $(G,k)$ to $(G_2,k_2)$,
the value of $r$ decreases by at least $(|H_2|+2)/3 \geq (|H'|+4)/3 \geq 2$
($|H_2|\geq |S_2| = |N(H')| = |C'|+|N(H') \setminus C'| \geq |C'|+3 = |H'|+2$).
Therefore, the branching vector in this case is at least $(1,2)$ and
the branching number is at most 1.6181.

We now describe algorithm \vctbaseb. The algorithm applies
the first applicable rule from the following rules.
\begin{Rule}\label{rule:terminate-1}
If $r < 0$, return `no'.
\end{Rule}

\begin{Rule}\label{rule:terminate-2}
If $V(G) = \emptyset$ return `yes'.
\end{Rule}

\begin{Rule}\label{rule:fold}
If there is a vertex $v$ with degree~2 whose two neighbors are not adjacent
and $|N(N(v))\setminus \{v\}| \leq 2$,
apply Rule~(C3) on $(G,k)$ with $C = \{v\}$ and $H = N(v)$,
and let $(G',k')$ be the resulting instance.
Return $\brancht{G'}{k'}{\emptyset}$.
\end{Rule}

Note that if Rule~(\currentrule) is applied, when moving from $(G,k)$ to
$(G',k')$, the value of $r$ decreases by $|H|-(|C|+|H|)/3 = \frac{1}{3}$.

\begin{Rule}\label{rule:no-fold}
If there is a vertex $v$ with degree~2 whose two neighbors are not adjacent,
return $\brancht{G}{k}{N(v),N(N(v))}$.
\end{Rule}
We now analyze the branching number of Rule~(\currentrule).
We can assume that the call to \branchtb\ executes line~\ref{branch:branch}
(if line \ref{branch:C1}, \ref{branch:C2} or~\ref{branch:C3} is executed then
no branching is performed, and if line~\ref{branch:recursion} is executed,
we already showed that the branching number is at most 1.6181).
By Lemma~\ref{lem:r-branch}, when procedure \branchtb\ processes the set
$S_1 = N(v)$, the decrease in $r$ is at least
$(|H_1|+2)/3 \geq (|S_1|+2)/3 = \frac{4}{3}$.
Additionally, when procedure \branchtb\ processes the set $S_2 = N(N(v))$,
the decrease in $r$ is at least $(|H_2|+2)/3 \geq (|S_2|+2)/3 \geq 2$.
Therefore, the branching vector is at least $(\frac{4}{3},2)$, and the branching
number is at most 1.5248.

\begin{Rule}\label{rule:triangle}
If there is a vertex $v$ with degree~2, return $\brancht{G}{k}{N(v)}$.
\label{rule:2nd-last}
\end{Rule}

\begin{Rule}\label{rule:3-regular}
If $G$ is 3-regular, select an arbitrary vertex $v$ and return
$\brancht{G}{k}{\{v\},N(v)}$.
\end{Rule}

We do not analyze the branching number of Rule~(\currentrule) since this rule is applied
at most once and therefore does not affect the time complexity of the algorithm.
The running time of algorithm \vctbaseb\ is $O^*(1.6181^r)$.
In order to handle non-connected graphs with maximum degree at most~3,
we use an algorithm \vctb\ that performs the following steps.
\begin{enumerate}
\item
If $G$ is connected return $\vctbase{G}{k}$.
\item
Let $G'$ be a connected component of $G$.
\item\label{vc:kprime}
For $k' = \lceil |V(G')|/3\rceil,\ldots,k-\lceil (n-|V(G')|)/3\rceil$,
if $\vctbase{G'}{k'}$ returns `yes', return $\vct{G-V(G')}{k-k'}$.
\item
Return `no'.
\end{enumerate}
Note that in line~\ref{vc:kprime} we have $k' \leq k-(n-|V(G')|)/3$ and
therefore $k'-|V(G')|/3 \leq k-n/3 = r$.
Thus, the time complexity of all the calls to \vctbaseb\ in line~\ref{vc:kprime}
is $O^*(1.6181^r)$.
It follows that the running time of algorithm \vctb\ is $O^*(1.6181^r)$.

We now describe an algorithm \vcbaseb\ for connected graphs with maximum
degree~4 and minimum degree at least~2.
The algorithm maintains the following invariants on the current graph:
\begin{inparaenum}[(1)]
\item\label{inv:min-deg-2}
The minimum degree is at least 2.
\item\label{inv:max-deg-4}
The maximum degree is~4.
\item\label{inv:not-4-regular}
Except for the initial instance, every connected component is not 4-regular.
\end{inparaenum}
Let \branchb\ be a procedure identical to \branchtb\ except that the calls
to \vctbaseb\ are replaced with calls to \vcbaseb.
Algorithm \vcbaseb\ uses Rules~(1)--(\ref{rule:2nd-last}) above,
where the calls to \branchtb\ are replaced with calls to \branchb,
and the following rules.

\begin{Rule}\label{rule:degree-3}
If the maximum degree of $G$ is at most 3, return $\vct{G}{k}$.
\end{Rule}

\begin{Rule}\label{rule:4-regular}
If $G$ is 4-regular, select an arbitrary vertex $v$ and return
$\branch{G}{k}{\{v\},N(v)}$.
\end{Rule}

\begin{Rule}\label{rule:non-IS}
If $v$ is a vertex with degree~3 such that there is an edge between
two neighbors $u,w$ of $v$, return
$\branch{G}{k}{N(v),N(x)}$,
where $x$ is the third neighbor of $v$.
\end{Rule}

By Lemma~\ref{lem:r-branch}, the branching vector of Rule~(\currentrule) is
at least $(\frac{5}{3},\frac{5}{3})$ (since $|N(v)| =  3$ and $|N(x)|\geq 3$),
and the branching number is at most 1.5158.

\begin{Rule}\label{rule:common}
If there is a vertex $v$ with degree~3 and a vertex $t$ such that
$|N(t) \cap N(v)| \geq 2$, return $\branch{G}{k}{N(v),\{v,t\}}$.
\end{Rule}

By Lemma~\ref{lem:r-branch}, the branching vector of Rule~(\currentrule) is
at least $(\frac{5}{3},\frac{4}{3})$,
and the branching number is at most 1.5906.

For the following rules note that there is a connected component of $G$
that contains at least one vertex with degree~3 and at least one vertex with
degree~4.
Therefore, there is a vertex $v$ with degree~3 that is adjacent to a vertex
with degree~4.
Denote the neighbors of $v$ by $z,u,w$, where $\deg(z) = 4$.
Note that since Rule~(\ref{rule:common}) cannot be applied,
the sets $N(z) \setminus \{v\}$, $N(u) \setminus \{v\}$, and
$N(w) \setminus \{v\}$ are pairwise disjoint.

\begin{Rule}\label{rule:uw-4}
If at least one vertex from $u,w$ has degree~4,
return $\branch{G}{k}{N(v), \{z\} \cup N(u) \cup N(w), N(z)}$.
\end{Rule}

By Lemma~\ref{lem:r-branch}, the branching vector of Rule~(\currentrule) is
at least $(\frac{5}{3},3,2)$ (since $|N(v)|=3$,
$|\{z\} \cup N(u) \cup N(w)| \geq 7$, and $|N(z)| = 4$).
We can improve the bound on the branching vector as follows.
Consider the processing of $S_3 = N(z)$ in procedure \branchb.
We consider two cases.
For the first case, assume that Rule~(D1) is not applied in
line~\ref{branch:apply-D0-D1} of procedure \branchb.
We claim that in this case Rule~(D0) is applied only on $z$.
Suppose conversely that the rule is applied on a vertex $x \neq z$.
By definition, $N(x) \subseteq N(z)$.
$x$ cannot be $u$ or $v$ since the sets $N(z) \setminus \{v\}$,
$N(u) \setminus \{v\}$, and $N(w) \setminus \{v\}$ are pairwise disjoint.
Therefore, $N(x) \subseteq N(z) \setminus \{v\}$.
The graph $G$ has minimum degree~3, and therefore $N(x) = N(z) \setminus \{v\}$.
It follows that Rule~(\ref{rule:common}) can be applied on $x$, a contradiction.
Thus, Rule~(D0) is applied only on $z$.
It follows that $H_3 = N(z)$ and $C_3 = \{z\}$.
Therefore, the decrease in $r$ is $4 - 5/3 = \frac{7}{3}$.
If Rule~(D1) is applied at least once, $|H_3| \geq |S_3|+1 = 5$.
By Lemma~\ref{lem:r-branch} the decrease in $r$ is
at least $(|H_3|+2)/3 \geq \frac{7}{3}$.
We obtained that the branching vector is at least  $(\frac{5}{3},3,\frac{7}{3})$
and the branching number is at most 1.6253.

\begin{Rule}\label{rule:uw-3}
Otherwise ($\deg(u) = \deg(w) = 3$), return $\branch{G}{k}{\{z\},N(z)}$.
\end{Rule}

By Lemma~\ref{lem:r-branch}, the branching vector of Rule~(\currentrule) is
at least $(\frac{2}{3},2)$ (since $|N(z)| = 4$).
We can improve the bound on the branching vector as follows.
Consider the processing of $S_1 = \{z\}$ in procedure \branchb.
Since $G$ has minimum degree~3 and $|S_1|=1$, we have that $G-S_1$ has
minimum degree at least~2. Therefore, Rules~(D0) and~(D1) are not applied by
procedure \branchb, so $G_1 = G-\{z\}$ and $k_1 = k-1$.
In the instance $(G_1,k_1)$, $v$ has degree~2 and its two neighbors $u,w$
are not adjacent (since Rule~(\ref{rule:non-IS}) could not be applied on
$(G,k)$).
Moreover, $|N_{G_1}(\{u,w\})\setminus \{v\}| = 4$
(since $N(u) \setminus \{v\}$ and $N(w)\setminus \{v\}$ are disjoint, and these
sets do not contain $z$).
Therefore, in the instance $(G_1,k_1)$ Rule~(\ref{rule:no-fold}) can be applied
on $v$.
We change the algorithm to force that in the instance $(G_1,k_1)$,
Rule~(\ref{rule:fold}) will be applied on $v$ (even though the requirement
$|N(N(v))\setminus \{v\}| \leq 2$ is not satisfied).
The application of Rule~(\ref{rule:fold}) on $v$ decreases the value of $r$ by
$\frac{1}{3}$.
Taking this decrease into account,
the branching vector of Rule~(\currentrule) is at least
$(\frac{2}{3}+\frac{1}{3},2)$ and the branching number is at most 1.6181.

We now show that the algorithm maintains the three invariants defined above.
Suppose that the current input graph $G$ to \vcbaseb\ satisfies the invariants.
We need to show that the graphs generated from $G$ on which \vcbaseb\ is
recursively called also satisfy the invariants.
This is straightforward for all the rules of the algorithm except Rule~(\ref{rule:uw-3}).
Consider the application of Rule~(\ref{rule:uw-3}) on $G$
and the application of Rule~(\ref{rule:fold}) on $G_1 = G-\{z\}$,
and let $G_2$ be the resulting graph.
We have $\deg_{G_2}(v^*) = |N_{G_1}(\{u,w\})\setminus \{v\}| = 4$.
All the other vertices in $G_2$ have degrees between 2 and 4.
Therefore, $G_2$ satisfies Invariant~(\ref{inv:min-deg-2}) and
Invariant~(\ref{inv:max-deg-4}).
At least one of the two vertices in $N(u)\setminus \{v\}$ has degree 3
in $G$ (otherwise Rule~(\ref{rule:uw-4}) can be applied on $u$).
This vertex also has degree~3 in $G_2$ (since $N(z)\setminus\{v\}$ and
$N(u)\setminus\{v\}$, and $N(w)\setminus\{v\}$ are pairwise disjoint)
and therefore $G_2$ satisfies Invariant~(\ref{inv:not-4-regular}).

We obtain that algorithm \vcbaseb\ solves the vertex cover problem on
connected graphs with maximum degree 4 and minimum degree at least 2 in
$O^*(1.6253^r)$ time.
To handle non-connected graph, we use an algorithm \vcb\ which is analogous to
algorithm \vctb.

\bibliographystyle{abbrv}
\bibliography{vc4}

\end{document}